\documentclass[11pt, envcountsect]{llncs}
\usepackage{fullpage}
\usepackage{amsmath,amssymb}
\usepackage{comment}

\newtheorem{protocol}{Protocol}

\renewcommand{\Pr}{\ensuremath\mathrm{Pr}}
\newcommand{\rar}[2]{\rule[-.1in]{0in}{.2in}\xrightarrow{\makebox[#1][c]
{$\textstyle #2 $}}}
\newcommand{\lar}[2]{\rule[-.1in]{0in}{.2in}\xleftarrow{\makebox[#1][c]
{$\textstyle #2 $}}}

\begin{document}

\title{Extended Combinatorial Constructions for Peer-to-peer User-Private Information Retrieval}
\author{C.M. Swanson \and D.R. Stinson\thanks{Research supported by NSERC grant 203114-11}}
\institute{David C. Cheriton School of Computer Science\\ University of Waterloo}
\date{}

\maketitle

\begin{abstract}We consider user-private information retrieval (UPIR), an interesting alternative to private information retrieval (PIR) introduced by Domingo-Ferrer et al. In UPIR, the database knows which records have been retrieved, but does not know the identity of the query issuer. The goal of UPIR is to disguise user profiles from the database. Domingo-Ferrer et al.\ focus on using a peer-to-peer community to construct a UPIR scheme, which we term P2P UPIR. In this paper, we establish a strengthened model for P2P UPIR and clarify the privacy goals of such schemes using standard terminology from the field of privacy research. In particular, we argue that any solution providing privacy against the database should attempt to minimize any corresponding loss of privacy against other users. We give an analysis of existing schemes, including a new attack by the database. Finally, we introduce and analyze two new protocols. Whereas previous work focuses on a special type of combinatorial design known as a configuration, our protocols make use of more general designs. This allows for flexibility in protocol set-up, allowing for a choice between having a dynamic scheme (in which users are permitted to enter and leave the system), or providing increased privacy against other users.
\end{abstract}

\section{Introduction}
We consider the case of a user who wishes to maintain privacy when requesting documents from a database. One existing method to address this problem is \emph{private information retrieval (PIR)}. In PIR, the content of a given query is hidden from the database, but the identity of the user making the query is not protected. In this paper, we focus on an interesting alternative to PIR dubbed \emph{user-private information retrieval (UPIR)}, introduced by Domingo-Ferrer et al.~\cite{DBWM09}. UPIR, however, is only nominally related to PIR, in that it seeks to provide privacy for users of a database. In UPIR, the database knows which records have been retrieved, but does not know the identity of the person making the query. The problem that we address, then, is how to disguise user profiles from the point of view of the database. 

We draw some of our terminology from Pfitzmann~\cite{PH10}. Here we understand \emph{anonymity} as the state of not being identifiable within a set of subjects, and the \emph{anonymity set} is the set of all possible subjects. By \emph{untraceable} queries from the point of view of the database, we mean that the database cannot determine that a given set of queries belongs to the same user. One interesting caveat, which is addressed below, is that a set of queries might be deemed to come from the same user based on the subject matter of those queries. If the subject matter of a given set of queries is esoteric or otherwise unique, the database (or some other adversary) can surmise that the identity of the source is the same for all (or most) queries in this set; we call such a set of queries \emph{linked}. In the case of linked queries, we wish to provide as much privacy as possible, in the sense that we wish the database to have no probabilistic advantage in guessing the identity of the source of a given set of linked queries. In this way, we can say the user making the linked queries still has \emph{pseudonymity}---his identity is not known.

With this terminology in mind, we might better explain UPIR as a method of database querying that is privacy-preserving and satisfies the following properties from the point of view of the database:
\begin{enumerate}
	\item For any given user $U_i$, some (large) subset of all users $\mathcal{U}$ is the query anonymity set for $U_i$;
	\item User queries are anonymous; 
	\item User queries are untraceable;
	\item Given a set of queries that is unavoidably traceable due to subject matter, the person making the query is protected by pseudonymity.
\end{enumerate}

In addition to these basic properties of user-privacy against a database, we may wish to provide user-privacy against other users. Ideally, a UPIR scheme would provide the same privacy guarantees against other users as against the database, but we will see that this usually cannot be attained in practice.

Previous work~\cite{DB08,DBWM09,SB10,SB11} has focused on the use of a P2P network consisting of various encrypted ``memory spaces'' (i.e., drop boxes), to which users can post their own queries, submit queries to the database and post the respective answers, and read answers to previously posted queries. That is, in the P2P UPIR setting, we have a cooperating community of users who act as proxies to submit each other's queries to the database. In particular, a class of combinatorial  designs known as \emph{configurations} have been suggested by Domingo-Ferrer,  Bras-Amor\'{o}s et al.~\cite{DB08,DBWM09,SB10,SB11} as a way to specify the structure of the P2P network. In this work, we focus on P2P UPIR and consider the application of other types of designs in determining the structure of the P2P network. We introduce new P2P UPIR protocols and explore the level of privacy guarantees our protocols achieve, both against the database and against other users.

\subsection{Our Contributions}
\label{subsec: Contributions}

The main contributions of our work are as follows.

\begin{itemize}
	\item We establish a strengthened model for P2P UPIR and clarify the privacy goals of such schemes using standard terminology from the field of privacy research.
	\item We provide an analysis of the protocol introduced by Domingo-Ferrer and Bras-Amor\'{o}s~\cite{DB08,DBWM09}, as well as its subsequent variations. In particular, we reconsider the choice to limit the designs used as the basis for the P2P UPIR scheme to configurations. We provide a new attack on user-privacy against the database, which we call the \emph{intersection attack}, to which the above protocol variations are vulnerable.
	\item We introduce two new P2P UPIR protocols (and variations on these), and give an analysis of the user-privacy these protocols provide, both against the database and against other users. Our protocols utilize more general designs and resist the intersection attack by the database. In particular, our protocols provide more flexibility in designing the P2P network. 
	\item We consider the possible trade-offs of using different types of designs in the P2P UPIR setting, both with respect to the overall flexibility of the scheme as well as user-privacy. Our protocols provide viable design choices, which can allow for a \emph{dynamic UPIR scheme} (i.e., one in which users are permitted to enter and leave the system), or provide increased privacy against other users. 
	\item We consider the problem of user-privacy against other users in detail. In particular, we relax the assumptions of previous work, by allowing users to collaborate outside the parameters of the P2P UPIR scheme; that is, we consider a stronger adversarial model than previous work. We analyze the ability of different types of designs to provide user-privacy against other users, and explore how well our protocol resists an intersection attack launched by a coalition of users on linked queries. Finally, we introduce methods to improve privacy against other users without compromising privacy against the database.
	\end{itemize}
	
We now give an outline of our paper. In Section~\ref{sec: P2P UPIR Model}, we give a model for P2P UPIR schemes and provide the relevant privacy goals. Section~\ref{sec: Background on Designs} provides background information on designs. We then review previous work in Section~\ref{sec: Previous Work: Using Configurations} and give attacks on these protocols in Section~\ref{subsec: Attacks}. We introduce our protocols in Section~\ref{sec: Using More General Designs} and give an analysis of the privacy guarantees our protocols provide against the database. In Section~\ref{sec: Privacy Against Other Users}, we analyze the ability of our protocols to provide user-privacy against other users and consider ways to improve this type of privacy. We conclude in Section~\ref{sec: Conclusion}.

\section{Our P2P UPIR Model}
\label{sec: P2P UPIR Model}

A \emph{P2P UPIR scheme} consists of the following players: a finite set of possible \emph{users} $\mathcal{U} = \{U_1, \ldots, U_v\}$, the \emph{target database} DB, and an \emph{external observer}, $O$. We assume all communication in a P2P UPIR scheme is encrypted, including communication between the users and DB.

In a basic P2P UPIR scheme, users have access to secure drop boxes known as \emph{memory spaces}. More precisely, a \emph{memory space} is an abstract (encrypted) storage space in which some subset of users can store and extract queries and query responses; the exact structure of these spaces is not specified. We let $\mathcal{S} = \{S_1,\ldots, S_b\}$ denote the set of memory spaces, and we let $K_i$ denote the (symmetric) key associated with $S_i$, for $1 \leq i \leq b$. We assume that encryption keys for memory spaces are only known to a given subset of users, as specified by the P2P UPIR protocol. For the sake of simplicity, we assume that these keys are initially distributed in a secure manner by some trusted external entity (not the database DB). However, the precise method by which these keys are distributed is not relevant to the results we prove in this paper. If two distinct users $U_i, U_j \in \mathcal{U}$ have access to a common memory space, then we say $U_i$ and $U_j$ are \emph{neighbors}. Similarly, the \emph{neighborhood} of a user $U_i$ is defined as the set of all neighbors of $U_i$, and it is denoted as $N(U_i)$.

When a user $U_i$ wishes to send a query $q$ to DB, we say $U_i$ is the \emph{source} of the query. Rather than sending the query directly to DB, $U_i$ writes an encrypted copy of $q$, together with a requested proxy $U_j$, to a memory space $S_\ell$. Here $U_j$ is the \emph{proxy} for $U_i$'s query $q$, and consequently $U_j$ must know the encryption key $K_\ell$ corresponding to the memory space $S_\ell$. The user $U_j$ decrypts the query, re-encrypts $q$ under a secret key shared with DB, say $K_{DB}^j$, and forwards this re-encrypted query $e_{K_{DB}^j}(q)$ to DB. DB sends back a response, which $U_j$ first decrypts, then re-encrypts under $K_\ell$ and records in the memory space $S_\ell$. We give a schematic of the information flow of a basic P2P UPIR scheme in Figure~\ref{information flow}.

\begin{figure}[!h]
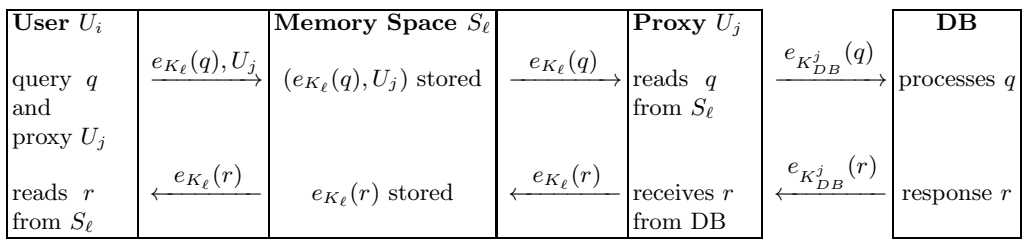
 \begin{center}
\caption{Schematic of Information Flow}
\vspace{.3cm}
\label{information flow}
\setlength{\tabcolsep}{2pt}
\begin{tabular}{|p{.63in}|c|c|c|p{.65in}|c|c|} \cline{1-5}\cline{7-7} 
		\textbf{User $U_i$} & &	\textbf{Memory Space $S_\ell$} & & \textbf{Proxy $U_j$} & &\textbf{DB}\\
		query $q$ \; \, and \; \; \; proxy $U_j$ & $\rar{1.3cm}{e_{K_\ell}(q), U_j}$ & $(e_{K_\ell}(q), U_j)$ stored & $\rar{1.3cm}{e_{K_\ell}(q)}$ & reads $q$ \; \, from $S_\ell$ & $\rar{1.3cm}{e_{K_{DB}^j}(q)}$ & processes $q$\\ 
		reads $r$ \; \, from $S_\ell$ & $\lar{1.3cm}{e_{K_\ell}(r)}$ & $e_{K_\ell}(r)$ stored & $\lar{1.3cm}{e_{K_\ell}(r)}$ &  receives $r$ \; \, from DB & $\lar{1.3cm}{e_{K_{DB}^j}(r)} $ &response $r$ \\
	\cline{1-5}\cline{7-7}
\end{tabular}
\end{center}
\end{figure}

\subsection{Attack Model}
\label{subsec: Attack Model}

We consider each type of player as a possible adversary $\mathcal{A}$. We assume that $\mathcal{A}$ has full knowledge of the P2P UPIR scheme specification, including any public parameters, as well as any secret information assigned to $\mathcal{A}$ as part of the P2P UPIR scheme. In addition, we assume $\mathcal{A}$ does not conduct traffic analysis. The following definition will be useful.

\begin{definition} Consider a set of one or more users $C$. \emph{The query sphere for $C$} is the set of memory spaces that $C$ can (collectively) access via the P2P UPIR scheme.
\end{definition}

In addition to the above, we make the following assumptions about each specific type of adversary $\mathcal{A}$:

\begin{itemize}
\item Suppose $\mathcal{A}$ is the database DB. As stated above, we assume that DB does not observe information being posted to or read from memory spaces. In addition, we assume that DB does not collaborate with any users and answers queries honestly. We note DB necessarily observes the content of all queries and the proxy of each query.

\item Suppose $\mathcal{A}$ consists of a user or a subset of colluding users $C \subset \mathcal{U}$. We assume users are honest-but-curious. Users in $C$ can communicate outside of the given P2P UPIR scheme and collaborate using joint information. The users of $C$ can see the content of any queries within $C$'s query sphere, but cannot identify the original source of these queries.

\item Suppose $\mathcal{A}$ is an external adversary $O$. An external observer $O$ can see the encrypted content of memory spaces. We consider the possibility of \emph{key leakage} as the main attack launched by $O$. This refers to a party gaining access to a memory space key outside of the P2P UPIR scheme specification (e.g., by social engineering or other means).
\end{itemize}

Although we do not specifically treat traffic analysis as an attack, we wish to avoid a trivial analysis of traffic entering and leaving a given memory space. That is, we assume that memory spaces have encryption and decryption capabilities, so that a user acting as a proxy may decrypt and re-encrypt a given query within its associated memory space, before forwarding the query to the database.

\subsection{Privacy and Adversarial Goals in P2P UPIR}
\label{subsec: Privacy in P2P UPIR}

In considering the privacy guarantees for a user $U_i$, we assume either the database DB or a group of other users may try to determine whether $U_i$ is the source of a given set of queries, or try to establish whether or not a given set of queries originates from the same source. We will need the following definition:

\begin{definition} We say two or more queries $q_1, q_2, \ldots$ are \emph{linked} if, given the subject matter, one can infer that the queries are likely to be from the same source.
\end{definition}

We also consider the possibility of an external observer $O$ gaining information that compromises the privacy of $U_i$. We recognize the following goals for $U_i$'s privacy:
\begin{itemize}
	\item \emph{Confidentiality}: the content of $U_i$'s queries is protected;
	\item \emph{Anonymity}: the identity of a query source is protected;
	\item \emph{Untraceability}: a user's query history cannot be reconstructed as having originated from the same user;
	\item \emph{Pseudonymity in the presence of linked queries}: given a set of linked queries, the identity of the source is protected.
\end{itemize}

We now analyze each type of adversary $\mathcal{A}$ with respect to the above privacy goals:

\begin{itemize}
	\item Suppose $\mathcal{A}$ is the database DB. We are not concerned with confidentiality against DB, but rather anonymity, untraceability, and pseudonymity in the presence of linked queries. The goal of the database is to create a profile of $U_i$. That is, the database would like to establish the  set of queries for which $U_i$ is the source. The database also attempts to trace user query histories; that is, DB would like to establish that a given set of queries came from the same source, even if DB cannot determine the identity of the source.
				
	\item Suppose $\mathcal{A}$ consists of a user or a subset of colluding users $C \subset \mathcal{U}$.  The coalition $C$ collaborates to try to determine the query history of another user $U_i \notin C$. Here we are interested in maintaining anonymity, untraceability, and pseudonymity in the presence of linked queries against $C$. We are also interested in maintaining confidentiality, in the sense that $C$ should not have access to the content of queries outside the query sphere for $C$. 
		
	\item Suppose $\mathcal{A}$ is an external adversary $O$. The goal of $O$ is to compromise both the confidentiality and the anonymity of $U_i$. External adversaries may try to compromise the encryption mechanism of the memory spaces. 
\end{itemize}

We are now almost ready to consider the P2P UPIR protocols of Domingo-Ferrer, Bras-Amor\'{o}s et al.~\cite{DB08,DBWM09}, as well as the subsequent modification of Stokes and Bras-Amor\'{o}s~\cite{SB10,SB11}. Both these protocols and ours, however, draw heavily from the field of combinatorial designs. In the next section, we introduce the requisite background knowledge on combinatorial designs.

\section{Background on Designs}
\label{sec: Background on Designs}
For a general reference on designs, we refer the reader to Stinson~\cite{S03}.

\begin{definition} A \emph{set system} is a pair $(X, A)$ such that the following are satisfied:
	\begin{enumerate}
		\item $X$ is a set of elements called \emph{points}, and
		\item $A$ is a collection (i.e., multiset) of nonempty proper subsets of $X$ called \emph{blocks}.
	\end{enumerate}
\end{definition}

In the rest of this section, we abuse notation by writing blocks in the form $abc$ instead of $\{a,b,c\}$.

\begin{definition} The \emph{degree} of a point $x \in X$ is the number of blocks containing $x$. If all points have the same degree, $r$, we say $(X,A)$ is \emph{regular (of degree $r$)}.
\end{definition}

\begin{definition} The \emph{rank} of $(X , A)$ is the size of the largest block. If all blocks contain the same number of points, say $k$, then $(X,A)$ is \emph{uniform (of rank $k$)}. Note that $k < v$.
\end{definition}

\begin{definition} A \emph{covering design} is a set system in which every pair of points occurs in at least one block.
\end{definition}

\begin{example} A covering design.\\
\[X = \{1, 2, 3, 4, 5, 6, 7\} \mbox{ and } A = \{13,23, 157, 124, 347, 356, 2567, 14567 \}.\]
\end{example}

\begin{definition}\label{PBD} A \emph{pairwise balanced design} (or \emph{PBD}) is a set system such that every pair of distinct points is contained in exactly $\lambda$ blocks, where $\lambda$ is a fixed positive integer. Note that any PBD is a covering design.
\end{definition}

\begin{example} A PBD with $\lambda = 2$.\\
\[X = \{1,2,3,4,5\} \mbox{ and } A = \{12, 25, 135, 145, 1234, 2345\}.\]
\end{example}

\begin{definition} Let $(X, A)$ be a regular and uniform set system of degree $r$ and rank $k$, where $|X| = v$ and $|A| = b$. We say $(X,A)$ is a \emph{$(v, b, r, k)$-1 design}.
\end{definition}

\begin{example} A $(5,5,3,3)$-1 design.\\
\[X = \{1,2,3,4,5\} \mbox{ and } A = \{123,451,234,512,345\}.\]
\end{example}

\begin{definition} A \emph{$(v, b, r, k, \lambda)$-balanced incomplete block design} (or \emph{BIBD}) is a $(v, b, r, k)$-1 design in which every pair of points occurs in exactly $\lambda$ blocks. Equivalently, a $(v, b, r, k, \lambda)$-BIBD is a PBD that is regular and uniform of degree $r$ and rank $k$.
\end{definition}

\begin{example} A $(10, 15, 6, 4, 2)$-BIBD.\\
\[X = \{0,1,2,3,4,5,6,7,8,9\}\] and
\begin{equation*}
\begin{split}
A &= \{0123,0147,0246,0358,0579,0689,1258,1369,\\
& \quad 1459,1678,2379,2489,2567,3478,3456\}.
\end{split}
\end{equation*}
\end{example}

\begin{definition} A \emph{$(v, b, r, k)$-configuration} is a $(v, b, r, k)$-1 design such that every pair of distinct points is contained in at most 1 block.
\end{definition}

\begin{example} A $(9,9,3,3)$-configuration.\\
\[X = \{1,2,3,4,5,6,7,8,9\} \mbox{ and } A = \{147, 258, 369, 159, 267, 348, 168, 249, 357\}.\]
\end{example}

\begin{remark} A $(v, b, r, k)$-configuration with $v = r(k-1) + 1$ is a $(v, b, r, k, 1)$-BIBD. A $(v, b, r, k, 1)$-BIBD with parameters of the form $(n^2 + n + 1, n^2 + n + 1, n+1, n+1, 1)$ is a \emph{finite projective plane of order $n$}.
\end{remark}

\begin{definition} A \emph{symmetric BIBD} is a BIBD in which $b = v$.
\end{definition}

\begin{remark} A projective plane is a symmetric BIBD.
\end{remark}

\begin{theorem}\emph{(Fisher's Inequality)} In any $(v,b,r,k,\lambda)$-BIBD, $b \geq v$.
\end{theorem}

\begin{theorem} \label{symmetric intersection}In a symmetric BIBD any two blocks intersect in exactly $\lambda$ points.
\end{theorem}

\subsection{P2P UPIR using Combinatorial Designs}
\label{subsec: P2P UPIR using Combinatorial Designs}

We model a P2P UPIR scheme using a combinatorial design. That is, we consider pairs $(\mathcal{U}, \mathcal{S})$ (where as before $|\mathcal{U}| = v$ and $|\mathcal{S}| = b$ ), such that each memory space, or \emph{block}, consists of $k$ users and each user, or \emph{point}, is associated with $r$ memory spaces. That is, we assume that the pair $(\mathcal{U}, \mathcal{S})$ is a $(v, b, r, k)$-1 design.

We can also view the $b$ memory spaces as points and define $v$ blocks, each of which contains the memory spaces to which a given user belongs. This yields the dual design $(\mathcal{S}, \mathcal{U})$, which is a $(b, v, k, r)$-1 design.

\section{Previous Work: Using Configurations}
\label{sec: Previous Work: Using Configurations}

We briefly review the P2P UPIR scheme proposed by Domingo-Ferrer et al.~\cite{DBWM09} and the proposed modification of Stokes and Bras-Amor\'{o}s~\cite{SB11}.  We fix a $(v,b,r,k)$-configuration $(\mathcal{U}, \mathcal{S})$. As before, we have a finite set of users $\mathcal{U} = \{U_1, \ldots, U_v\}$, a database DB, and a finite set of memory spaces $\mathcal{S} = \{S_1, \ldots, S_b\}$.

Each user has access to $r$ memory spaces, and each memory space is accessible to $k$ users. Each memory space is encrypted via a symmetric encryption scheme; for each memory space, only the $k$ users assigned to that memory space are given the key. The following protocol~\cite{DBWM09} assumes the user $U_i$ has a query to submit to the database:

\begin{protocol}\emph{Domingo-Ferrer--Bras-Amor\'{o}s--Wu--Manj\'{o}n (DBWM) Protocol}\\
\label{DBWM protocol}  We fix a $(v,b,r,k)$-configuration.
\begin{enumerate}
	\item  The user $U_i$ randomly selects a memory space $S_\ell$ to which he has access
	\item  The user $U_i$ decrypts the content on the memory sector $S_\ell$ using the
corresponding key. His behavior is then determined by the content on the memory space as follows:
 		\begin{enumerate}
 			\item The content is garbage. Then $U_i$ encrypts his query and records it in $S_\ell$.
 			\item The content is a query posted by another user. Then $U_i$ forwards the query to the database and awaits the answer. When $U_i$ receives the answer, he encrypts it and records it in $S_\ell$. He then restarts the protocol with the intention to post his query.
 			\item \label{DBWM version} The content is a query posted by the user himself. Then $U_i$ does not forward the query to the database. Instead $U_i$ restarts the protocol with the intention to post his query.
 			\item The content is an answer to a query posted by another user. Then $U_i$ restarts the protocol with the intention to post his query;
			\item The content is an answer to a query posted by the user himself. Then $U_i$ reads the query answer and erases it from the memory space. Subsequently $U_i$ encrypts his new query and records it in $S_\ell$.
		\end{enumerate}
\end{enumerate}
\end{protocol}

The modification proposed by Stokes and Bras-Amor\'{o}s~\cite{SB11} replaces~\ref{DBWM version} as follows:
\begin{protocol} \emph{DBWM--Stokes (DBWMS) Protocol}
\label{DBWMS Protocol}
	\begin{enumerate}
		\item[$2'$(c).] If the content is a query posted by the user himself, then $U_i$ forwards the query to the database with a specified probability $p$. If $U_i$ forwards the query to the database, he records the answer in $S_\ell$. The user $U_i$ restarts the protocol with the intention to post his current query.
	\end{enumerate}
\end{protocol}

\begin{remark}
This protocol is ambiguous as stated by Stokes and Bras-Amor\'{o}s. The intent of the system is that users periodically run the protocol with ``garbage'' queries, in this way collecting the answers to their previous queries.
\end{remark}

Stokes and Bras-Amor\'{o}s~\cite{SB10,SB11} argue that the finite projective planes are the optimal configurations to use for P2P UPIR. Their argument is that privacy against the database is an increasing function of $r(k-1)$, since there are $r(k-1)$ users in the anonymity set of any given user $U_i$. That is, the query profile of $U_i$ is diffused among $r(k-1)$ other users in the neighborhood of $U_i$. Now, since $r(k-1) \leq v - 1$ in a configuration, the authors consider configurations satisfying $r(k-1) = v-1$, which yield the finite projective planes. In our protocols, introduced in Section~\ref{sec: Using More General Designs}, we also have neighborhoods of maximum size, without limiting ourselves to configurations. We also ensure that the database DB has no advantage in guessing the identity of the source of any given query.

\subsection{Attacks}
\label{subsec: Attacks}

We consider the privacy properties of the DBWM and DBWMS protocols with respect to the database, before offering an improved protocol in Section~\ref{sec: Using More General Designs}. We fix a $(v,b,r,k)$-configuration, where $v$ is the number of users and $b$ is the number of memory spaces. We associate a block with each memory space, where the block consists of the users that have access to the memory space. 

The weakness of the DBWM and DBWMS protocols lies in the possibility of a user's query history being identifiable as originating from one user. That is, if a series of queries is on some esoteric subject, the adversary (such as the database) can surmise that the source of these queries is the same. As before, we refer to such queries as \emph{linked}.

Stokes and Bras-Amor\'{o}s~\cite{SB11} noticed a weakness in the DBWM protocol when a projective plane is used as the configuration, that is, when $v = r(k-1) + 1$. In this case, each user $U_i$ has a neighborhood consisting of all other users. Then, given a large enough set of linked queries, the only user who never submits one of these linked queries is the source, $U_i$. That is, the database can eventually identify $U_i$ as the source. Subsequent to our research, Stokes and Bras-Amor\'{o}s~\cite{SB11ws} noted this weakness applies more generally to $(v,k,1)$-BIBDs.

We introduce another type of attack, which we call the \emph{intersection attack}. This attack only applies to configurations satisfying $v > r(k-1)+1$, as it requires that all users have neighborhoods of cardinality less than $v - 1$. The idea behind the intersection attack is that, given a query $q_1$ submitted by proxy $U_j$, an attacker can, by analyzing the neighborhood of $U_j$, compute a list of possible sources $Q_1$. If the attacker has access to a set of linked queries $q_1, q_2, \ldots, q_n$, and the neighborhoods of these users do not consist of all users in the system, the intersection of the possible source sets $Q_1, Q_2, \ldots, Q_n$ can perhaps identify the source (or narrow down the list of possible sources). We demonstrate this attack in the following example.

\begin{example}
Suppose $v=12$ and $b=8$ and we have the following blocks (memory spaces):
\[ 
\begin{tabular}{llll}
$\{U_1, U_2,U_3\}$ & $\{U_4,U_5,U_6\}$ & $\{U_7,U_8,U_9\}$ & $\{U_{10},U_{11},U_{12}\}$ \\
$\{U_1, U_4, U_7\}$ & $\{U_2, U_5, U_{10}\}$ & $\{U_3, U_8, U_{11}\}$ & $\{U_6, U_9, U_{12}\}$
\end{tabular}
\]

Note this is a $(12, 8, 2, 3)$-configuration. We consider the DBWM protocol here; that is, we assume that the proxy of a given query is always different from the source of the query. Now suppose three queries are transmitted from users $U_2, U_{11}$, and $U_8$. 
\begin{itemize}
\item If the proxy is $U_2$, then the source $U_i \in \{U_1,U_3,U_5,U_{10}\}$.
\item If the proxy is $U_{11}$, then the source $U_i \in \{U_3,U_8,U_{10},U_{12}\}$.
\item If the proxy is $U_8$, then the source $U_i \in \{U_3,U_7,U_9,U_{11}\}$.
\end{itemize}
Suppose that the subject of the queries is similar, so it can be inferred that the source of the three queries is probably the same user. Then it is easy to identify the source of the queries: \[ U_i \in \{U_1,U_3,U_5,U_{10}\} \cap \{U_3,U_8,U_{10},U_{12}\} \cap \{U_3,U_7,U_9,U_{11}\},\] so $U_i = U_3$. Clearly, user-privacy with respect to the database is not achieved here.
\end{example}

We do not claim that the above-described attack will always work for any configuration; it is easy to come up with examples where the attack does not work.  For example, suppose that $N(U_i) \cup \{U_i\} = N(U_j) \cup \{U_j\}$ for two distinct users $U_i$ and $U_j$. Then it would be impossible for DB to determine whether $U_i$ or $U_j$ is the source of a sequence of linked queries. Independently of this research, Stokes and Bras-Amor\'{o}s~\cite{SB11ws} noted that by choosing the configuration carefully, it is possible to ensure the neighborhood-to-user mapping is not unique, and to guarantee a specified lower bound on the number of possible users for a given neighborhood.

Observe that the intersection attack is not useful when one uses a finite projective plane as the configuration and users are allowed to submit their own queries. This follows because, at each stage of the intersection attack, the set of possible sources includes all users in the set system. In the next section, we formalize this observation and discuss the use of more general types of designs in P2P UPIR protocols that resist the intersection attack in a very strong sense.

\section{Using More General Designs}
\label{sec: Using More General Designs}

As observed in Section~\ref{sec: Previous Work: Using Configurations}, in order to achieve user-privacy with respect to the database, we need to allow users to sometimes transmit their own queries. We suggest a different solution to the problem than that given by Bras-Amor\'{o}s et al., however. In particular, we see no reason to limit the P2P network topology to configurations. Bras-Amor\'{o}s et al.\ indicate use of configurations as a method to increase service availability and decrease the number of required keys. Indeed, configurations were proposed as key rings in wireless sensor networks by Lee and Stinson~\cite{LS05} due to memory constraints of sensor nodes. However, storage constraints are not so much an issue in P2P UPIR. We therefore consider the possibility of using other types of designs.

We will make use of memory spaces that ``balance'' proxies for every source. We suggest to use a balanced incomplete block design (BIBD) for the set of memory spaces. We will show that these designs provide optimal resistance against the intersection attack.

Our scheme also differs from DBWMS in the treatment of proxies. In the previous schemes, the identity of a proxy was not specified by the source. Queries were simply forwarded to the database by whichever user had most recently checked the corresponding memory space. We propose that each source designates the proxy for each query. This enables us to balance the proxies for each possible source, thereby providing ``perfect'' anonymity with respect to the database. Moreover, we do not assume that each memory space holds only a single query; rather, we assume that memory spaces are capable of storing multiple queries.

\begin{protocol}\label{proxy-designated version 1}\emph{Proxy-designated BIBD Protocol (Version 1)}\\ We fix a $(v,b,r,k,\lambda)$-BIBD. To submit a query, a user $U_i$ uses the following steps:
\begin{enumerate}
\item \label{self-submission} With probability $1/v$, user $U_i$ acts as his own proxy and transmits his own query to the DB.
\item Otherwise, user $U_i$ chooses uniformly at random one of the $r$ memory spaces with which he is associated, say $S_{\ell}$, and then he chooses uniformly at random a user $U_j \in S_{\ell} \backslash \{U_i\}$. Finally, user $U_i$ requests that user $U_j$ act as his proxy using the memory space $S_{\ell}$.
\end{enumerate}
\end{protocol}

\begin{protocol}\label{proxy-designated version 2}\emph{Proxy-designated BIBD Protocol (Version 2)}\\We fix a $(v,b,r,k,\lambda)$-BIBD. To submit a query, a user $U_i$ uses the following steps:
\begin{enumerate}
\item With probability $1/v$, user $U_i$ chooses to act as his own proxy. User $U_i$ then writes the query uniformly at random to one of the $r$ memory spaces with which he is associated, and transmits his own query to DB.
\item Otherwise, user $U_i$ chooses uniformly at random one of the $r$ memory spaces with which he is associated, say $S_{\ell}$, and then he chooses uniformly at random a user $U_j \in S_{\ell} \backslash \{U_i\}$. Finally, user $U_i$ requests that user $U_j$ act as his proxy using the memory space $S_{\ell}$.
\end{enumerate}
\end{protocol}

\begin{remark} We note that Protocol~\ref{proxy-designated version 2} differs from Protocol~\ref{proxy-designated version 1} only in the first step.
\end{remark}

\begin{remark} We assume users check memory spaces regularly and act as proxies as requested within a reasonable time interval.
\end{remark}

\begin{remark}\label{assumption} We make the assumption that, when a source $U_i$ requests $U_j$ to be his proxy, everyone in the associated memory space knows that this request has been made, but no one (except for $U_i$) knows the identity of the source.
\end{remark}

\begin{remark} The choice between Protocol~\ref{proxy-designated version 1} and Protocol~\ref{proxy-designated version 2} impacts the amount of privacy the scheme provides against other users. This will be discussed in Section~\ref{sec: Privacy Against Other Users}.
\end{remark}

We analyze the situation from the point of view of the database. For the rest of the paper, we let variables $\mathbf{S}, \mathbf{P}, \mathbf{M}$ be random variables for source, proxy, and memory space, respectively.

\begin{theorem} 
From the point of view of the database, the Proxy-designated BIBD Protocols (Protocols~\ref{proxy-designated version 1}~and~\ref{proxy-designated version 2}) satisfy $\Pr[\mathbf{S}= U_i|\mathbf{P} = U_j] = \Pr[\mathbf{S}=U_i]$ for all $U_i, U_j \in \mathcal{U}$. 
\end{theorem}

\begin{proof}
First, the schemes ensure that $\Pr[\mathbf{P}=U_j|\mathbf{S}=U_i] = \frac 1v$ for all $U_i,U_j$. To see this, first note that $U_i$ will pick himself as the source with probability $\frac 1 v$. In Protocol~\ref{proxy-designated version 1}, $U_i$ will then submit his query directly to the database. In Protocol~\ref{proxy-designated version 2}, $U_i$ will pick one of the $r$ memory spaces with which he is associated uniformly at random and then act as his own proxy. So in both cases, we have \[\Pr[ \mathbf{P} = U_i | \mathbf{S} = U_i] = \frac 1 v.\]

Then in both protocols, with probability $\frac{v-1}{v}$, user $U_i$ will pick a memory space $S_\ell$ (with $U_i \in S_\ell$) uniformly at random, followed by a proxy $U_j$ associated with $S_\ell$. 
The probability that a fixed $U_j$ with $i \neq j$ will act as proxy can be computed as follows. 

For $i \neq j$, we have

\begin{eqnarray*}
\Pr[\mathbf{P}=U_j | \mathbf{S} = U_i] & = & \frac{v-1}{v} \sum_{S_\ell :U_i, U_j \in S_\ell} \Pr[\mathbf{M} = S_\ell]\Pr[\mathbf{P}=U_j| \mathbf{M} = S_\ell ]\\
& = & \frac{v-1}{v} \sum_{S_\ell : U_i, U_j \in S_\ell} \frac{1}{r(k-1)} =  \left(\frac{v-1}{v}\right)\left(\frac {\lambda}{r(k-1)} \right)= \frac 1 v.
\end{eqnarray*}

Similarly, we can see that $\Pr[\mathbf{P} = U_j] = 1/v$ for all $U_j \in \mathcal{U}$:

\begin{eqnarray*}
\Pr[\mathbf{P}=U_j ] & = & \sum_{S_\ell : U_j \in S_\ell} \Pr[\mathbf{M} = S_\ell ]\Pr[\mathbf{P}=U_j| \mathbf{M} = S_\ell]  = \frac{r}{bk} = \frac 1 v.
\end{eqnarray*}

Now we have 
\[\Pr[\mathbf{S}= U_i|\mathbf{P} = U_j] = \frac{\Pr[\mathbf{P}=U_j|\mathbf{S}=U_i]\Pr[\mathbf{S} = U_i]}{ \Pr[\mathbf{P} = U_j]} = \Pr[\mathbf{S}=U_i]\] so the identity of the proxy gives no information about the identity of the source. \end{proof}

We observe that this analysis is independent of any computational assumptions, so the security is unconditional. Since we have achieved a perfect anonymity property, it follows that no information is obtained by analyzing linked queries.

\begin{example}
\label{projective plane}
To illustrate, consider a projective plane of order 2 with the following blocks:
\[ 
\begin{tabular}{llll}
$\{U_1,U_2,U_3\}$ & $\{U_1,U_4,U_5\}$ & $\{U_1, U_6, U_7\}$ &  $\{U_2,U_4,U_6\}$ \\
$\{U_2,U_5,U_7\}$ & $\{U_3,U_4,U_7\}$ & $\{U_3,U_5,U_6\}$ &
\end{tabular}
\]

We note that this is a $(7, 3, 3, 3, 1)$-BIBD. Suppose that the first query uses block $\{U_2,U_4,U_6\}$ with proxy $U_4$, and the second query uses block $\{U_2,U_5,U_7\}$ with proxy $U_2$. From the first query, DB knows that one of three blocks were used: $\{U_1, U_4, U_5\}$, $\{U_2,U_4,U_6\}$, or $\{U_3,U_4,U_7\}$. However, $\Pr[S=U_i | P=U_4] = \Pr[S=U_i]$  for all possible sources $U_i$, so DB has no additional information about the identity of the source, given that $P=U_4$.  From the second query, DB knows that one of three blocks were used: $\{U_1,U_2,U_3\}$, $\{U_2,U_4,U_6\}$, or $\{U_2,U_5,U_7\}$. Again, $\Pr[S=U_i | P=U_2] = \Pr[S=U_i]$ for all possible sources $U_i$, so DB has no additional information about the identity of the source, given that $P=U_2$. So even if DB suspects that both queries came from the same source, he has no way to identify the source.
\end{example}

\subsection{Extensions}
\label{subsec: Extensions}

We can consider using less structured designs than BIBDs, such as pairwise balanced designs or covering designs. It turns out that we can still achieve perfect anonymity with respect to DB, because our anonymity argument remains valid provided that $\Pr[P = U_j | S = U_i] = \frac1v$ for all $U_i, U_j \in \mathcal{U}$.

We next give a generalized protocol based on an arbitrary covering design. That is, we do not require constant block size $k$ or constant replication number $r$.

\begin{protocol}\label{covering design version 1}\emph{Proxy-designated Covering Design Protocol (Version 1)}\\
We fix a covering design. To submit a query, a user $U_i$ performs the following steps:
\begin{enumerate}
\item \label{pick proxy}User $U_i$ chooses the designated proxy $U_j$ uniformly at random. If $U_i = U_j$, then $U_i$ submits his query directly to DB and skips Step~\ref{pick space}.
\item \label{pick space} If $U_i \neq U_j$, then user $U_i$ chooses uniformly at random one of the memory spaces that contains both $U_i$ and $U_j$, say $S_{\ell}$. Then $U_i$ requests that user $U_j$ act as his proxy using memory space $S_{\ell}$.
\end{enumerate}
\end{protocol}

\begin{protocol}\label{covering design version 2}\emph{Proxy-designated Covering Design Protocol (Version 2)}\\
We fix a covering design. To submit a query, a user $U_i$ performs the following steps:
\begin{enumerate}
\item \label{pick proxy2}User $U_i$ chooses the designated proxy $U_j$ uniformly at random. (The user $U_i$ may choose himself as the proxy $U_j$.)
\item \label{pick space 2}User $U_i$ chooses uniformly at random one of the memory spaces that contains both $U_i$ and $U_j$, say $S_{\ell}$. Then $U_i$ requests that user $U_j$ act as his proxy using memory space $S_{\ell}$.
\end{enumerate}
\end{protocol}

\begin{remark} If the covering design is a BIBD, then Protocol~\ref{covering design version 1} is equivalent to Protocol~\ref{proxy-designated version 1} and  Protocol~\ref{covering design version 2} is equivalent to Protocol~\ref{proxy-designated version 2}.
\end{remark}

\begin{remark}
We must have a covering design to ensure that a suitable memory space $S_\ell$ always exists in Step~\ref{pick space} of Protocols~\ref{covering design version 1}~and~\ref{covering design version 2}.
\end{remark}

\begin{remark}\label{assumption covering design}
As in Protocols~\ref{proxy-designated version 1}~and~\ref{proxy-designated version 2}, we assume users check memory spaces regularly, and act as proxies as requested within a reasonable time interval. We also assume, as before, that when source $U_i$ requests that $U_j \neq U_i$ be his proxy, everyone in the associated memory space knows that this request has been made, but no one (except for $U_i$) knows the identity of the source.
\end{remark}

\begin{theorem} From the point of view of the database, for a given query, the Proxy-designated Covering Design Protocols (Protocols~\ref{covering design version 1}~and~\ref{covering design version 2}) satisfy $\Pr[\mathbf{S}= U_i|\mathbf{P} = U_j] = \Pr[\mathbf{S}=U_i]$ for all $U_i, U_j \in \mathcal{U}$. 
\end{theorem}

\begin{proof}Step~\ref{pick proxy} of both Protocol~\ref{covering design version 1} and Protocol~\ref{covering design version 2} ensures that $\Pr[\mathbf{P} = U_j | \mathbf{S} = U_i] = \frac1v$ for all $U_i, U_j \in \mathcal{U}$. Similarly, we can see that $\Pr[\mathbf{P}=U_j] = \frac1v$ for all $U_j$. We once again have \[\Pr[\mathbf{S}= U_i|\mathbf{P} = U_j] = \frac{\Pr[\mathbf{P}=U_j|\mathbf{S}=U_i]\Pr[\mathbf{S} = U_i]}{ \Pr[\mathbf{P} = U_j]} = \Pr[\mathbf{S}=U_i]\] so the identity of the proxy gives no information about the identity of the source. \end{proof}

As before, we observe that this analysis is independent of any computational assumptions, so the security is unconditional. Since we have achieved a perfect anonymity property, no information is obtained by analyzing linked queries.

\subsection{Dynamic P2P UPIR Schemes}
One benefit of using less structured designs than BIBDs is that the scheme can be \emph{dynamic}. That is, we can add and remove users, which allows greater flexibility in practice. 

To delete a user $U_i$ from Protocols~\ref{covering design version 1}~and~\ref{covering design version 2}, we simply remove $U_i$ from all the memory spaces with which he is associated. To avoid $U_i$ from reading any more queries written to these memory spaces, we also need a rekeying mechanism to update the associated keys. The same external entity that distributed the initial set of keys could be responsible for rekeying. The end result is a covering design with one fewer users than before. 

To add a user $U_{\operatorname{new}}$ in Protocols~\ref{covering design version 1}~and~\ref{covering design version 2}, we may use the following method. We first find  $\mathcal{M} = \{S_{h_1},\ldots, S_{h_\ell}\} \subseteq \mathcal{S}$ such that $S_{h_1} \cup \cdots \cup S_{h_\ell} = \mathcal{U}$. That is, we need a set of memory spaces whose union contains all current users. A greedy algorithm could be used to accomplish this task, although the resultant set $\mathcal{M}$ would likely not be optimal. Indeed, finding the minimum such set is NP-hard. (This is the minimum cover problem, which is problem SP5 in Garey and Johnson~\cite{GJ79}.)

Once we have identified a suitable set $\mathcal{M}$, we simply add $U_{\operatorname{new}}$ to each memory space in $\mathcal{M}$, and give $U_{\operatorname{new}}$ the associated keys. In addition, we need a mechanism by which to inform all users of $U_{\operatorname{new}}$'s presence in the scheme. The resulting set system is still a covering design---one which contains one more user than before.

\section{Privacy Against Other Users}
\label{sec: Privacy Against Other Users}
In this section, we consider our Protocols~\ref{proxy-designated version 1},~\ref{proxy-designated version 2},~\ref{covering design version 1},~and~\ref{covering design version 2} in the context of analyzing user-privacy against other users. We remind the reader of Remarks~\ref{assumption}~and~\ref{assumption covering design}: we assume that when a source $U_i$ requests that $U_j$ be his proxy, everyone in the associated memory space knows that this request has been made, but no one (except for $U_i$) knows the identity of the source. We now analyze the privacy of a given user relative to other users of the scheme. As we will see, if we wish to provide privacy against other users, a design that has more structure than a general covering design becomes useful. In particular, we will observe that the use of a regular PBD (see Definition~\ref{PBD}) in Protocols~\ref{covering design version 1}~and~\ref{covering design version 2} is desirable.

It is helpful to begin with an example: 

\begin{example}\label{example privacy against other users} 
Consider the projective plane from Example~\ref{projective plane} and suppose we use Protocol~\ref{proxy-designated version 1}. Suppose that user $P=U_4$ is requested to make a query in memory space $\{U_1, U_4, U_5\}$ by source $S=U_1$.
User $U_4$ knows that the source must be $U_1$ or $U_5$ (since he did not make the request himself). User $U_5$, however, knows that the source must be $U_1$ because 
\begin{enumerate}
\item $U_5$ did not make the request himself, and 
\item $U_4$ would not post a request to himself to transmit a query---he would just go ahead and transmit it himself.
\end{enumerate}
\end{example}

We can generalize the concept from Example~\ref{example privacy against other users}. Observe that in Protocols~\ref{proxy-designated version 1}~and~\ref{covering design version 1}, the requested proxy can rule out one possible source, and anyone else in the memory space (who is not the source) can rule out two possible sources. If we consider Protocols~\ref{proxy-designated version 2}~and~\ref{covering design version 2}, then users can rule out only one possible source (namely, themselves). That is, Protocols~\ref{proxy-designated version 2}~and~\ref{covering design version 2} improve the information theoretic privacy guarantees of the scheme with respect to the viewpoint of other users. However, we remark that in these versions, when a source acts as his own proxy, other users associated with the chosen memory space can see the content of the query. In Protocols~\ref{proxy-designated version 1}~and~\ref{covering design version 1}, if a user $U_i$ is both the source and proxy of a given query, then $U_i$ is the only user who sees the content of that query. Hence it may still be desirable to use Protocols~\ref{proxy-designated version 1}~and~\ref{covering design version 1}, if additional confidentiality is required.

An interesting related question is, when a particular user $U_t$ sees a query $q$ posted to the memory space $S_\ell$ that is not his own, whether or not $U_t$ has a probabilistic advantage in guessing the source of $q$. The following theorems show that, in order to minimize any such advantage, it is helpful to use a regular PBD in our protocols.

\begin{theorem}\label{PBD theorem prot 5}
Let $(X, A)$ be a regular PBD of degree $r$. Assume $(X, A)$ is used in the Proxy-designated Covering Design Protocol (Protocol~\ref{covering design version 1}) and assume that $\Pr[\mathbf{S}=U_i] = \frac1v$ for all $U_i \in \mathcal{U}$. Suppose $U_t \in S_\ell$ sees a query $q$ posted to $S_\ell$ that is not his own. Then, from the point of view of $U_t$, for a given query $q$ and $U_i, U_j \in S_\ell$ such that $i \neq t$, it holds that 

\[\Pr[\mathbf{S} = U_i | \mathbf{M} = S_\ell, \mathbf{P}= U_j] = 
		\left\{
			\begin{array}{rl}
 				0 & \mbox{ if } i=j\\
 				\frac{1}{|S_\ell|-2} & \mbox{ if } i \neq j.
 			\end{array} \right. \]
\end{theorem}

\begin{proof}
We first note that the protocol definition ensures that when $i = j$, we have $\Pr[\mathbf{S} = U_i | \mathbf{M} = S_\ell, \mathbf{P}= U_j] = 0$.

We now consider the case $i \neq j$. We set $\lambda_{ij} = |\{S_q | U_i, U_j \in S_q\}| = \lambda$. Thus, we have
\begin{eqnarray*}
\Pr[\mathbf{M} = S_\ell, \mathbf{P}=U_j | \mathbf{S} = U_i]
& = & \Pr[\mathbf{M} = S_\ell| \mathbf{P}=U_j, \mathbf{S} = U_i]\Pr[\mathbf{P}=U_j| \mathbf{S} = U_i]\\
 & = & \frac{1}{\lambda_{ij}v} =  \frac{1}{\lambda v}.
\end{eqnarray*}

Then because $i \neq t, j$, we have $\Pr[\mathbf{S} = U_i] = \frac{1}{v-2}$ and
\begin{eqnarray*}
\Pr[\mathbf{S} = U_i | \mathbf{M} = S_\ell, \mathbf{P}= U_j] & = & \frac{\Pr[\mathbf{S}=U_i] \Pr[\mathbf{M} = S_\ell, \mathbf{P} = U_j | \mathbf{S} = U_i]}{\Pr[\mathbf{M} = S_\ell, \mathbf{P} = U_j]}\\
& = & \frac{\Pr[\mathbf{S}=U_i] \Pr[\mathbf{M} = S_\ell, \mathbf{P} = U_j | \mathbf{S} = U_i]}{\sum_{\substack{U_h \in S_\ell\\ h \neq t,j}} \Pr[\mathbf{S} = U_h ]\Pr[\mathbf{M}=S_\ell, \mathbf{P}=U_j | \mathbf{S} = U_h]}\\
& = & \frac{\frac{1}{v(v-2)\lambda}}{\sum_{\substack{U_h \in S_\ell\\ h \neq t, j}} \frac{1}{v(v-2)\lambda
}}
 =  \frac{1}{|S_\ell|-2},
\end{eqnarray*}

as desired. \end{proof}

\begin{theorem}\label{PBD theorem}
Let $(X, A)$ be a regular PBD of degree $r$. Assume $(X, A)$ is used in the Proxy-designated Covering Design Protocol (Protocol~\ref{covering design version 2}) and assume that $\Pr[\mathbf{S}=U_i] = \frac1v$ for all $U_i \in \mathcal{U}$. Suppose $U_t \in S_\ell$ sees a query $q$ posted to $S_\ell$ that is not his own. Then, from the point of view of $U_t$, for a given query $q$ and $U_i, U_j \in S_\ell$ such that $i \neq t$, it holds that 

\[\Pr[\mathbf{S} = U_i | \mathbf{M} = S_\ell, \mathbf{P}= U_j] = 
		\left\{
			\begin{array}{rl}
 				\frac{\lambda}{\lambda + r(|S_\ell| -2)} & \mbox{ if } i=j\\
 				\frac{r}{\lambda+r(|S_\ell|-2)} & \mbox{ if } i \neq j.
 			\end{array} \right. \]
\end{theorem}

\begin{proof}
We first calculate $\Pr[\mathbf{M} = S_\ell, \mathbf{P}=U_j | \mathbf{S} = U_i]$, where $U_i \in S_\ell$. We again set $\lambda_{ij} = |\{S_q | U_i, U_j \in S_q\}|$. Since $(X, A)$ is a PBD of degree $r$, we have 
\[\lambda_{ij} = \left\{ \begin{array}{rl} r & \mbox{ if } i = j \\ \lambda &\mbox { if } i \neq j. \end{array} \right.\]

We have
\begin{eqnarray*}
\Pr[\mathbf{M} = S_\ell, \mathbf{P}=U_j | \mathbf{S} = U_i]
& = & \Pr[\mathbf{M} = S_\ell| \mathbf{P}=U_j, \mathbf{S} = U_i]\Pr[\mathbf{P}=U_j| \mathbf{S} = U_i]\\
& = & \left(\frac{1}{\lambda_{ij}}\right)\left(\frac 1v \right).
\end{eqnarray*}

Then because $i \neq t$, we have $\Pr[\mathbf{S} = U_i] = \frac{1}{v-1}$ and
\begin{eqnarray*}
\Pr[\mathbf{S} = U_i | \mathbf{M} = S_\ell, \mathbf{P}= U_j] & = & \frac{\Pr[\mathbf{S}=U_i] \Pr[\mathbf{M} = S_\ell, \mathbf{P} = U_j | \mathbf{S} = U_i]}{\Pr[\mathbf{M} = S_\ell, \mathbf{P} = U_j]}\\
& = & \frac{\Pr[\mathbf{S}=U_i] \Pr[\mathbf{M} = S_\ell, \mathbf{P} = U_j | \mathbf{S} = U_i]}{\sum_{\substack{U_h \in S_\ell \\ h \neq t}} \Pr[\mathbf{S} = U_h ]\Pr[\mathbf{M}=S_\ell, \mathbf{P}=U_j | \mathbf{S} = U_h]}\\
& = & \frac{\frac{1}{v(v-1)\lambda_{ij}}}{\sum_{\substack{U_h \in S_\ell\\ h \neq t}} \frac{1}{v(v-1)\lambda_{hj}}}\\
& = & \frac{1}{\lambda_{ij}\left(\frac 1r + \frac{|S_\ell|-2}{\lambda}\right)},
\end{eqnarray*}

which yields the desired result. 
\end{proof}

\begin{remark}
Theorems~\ref{PBD theorem prot 5}~and~\ref{PBD theorem} apply to Protocols~\ref{proxy-designated version 1}~and~\ref{proxy-designated version 2}, respectively, since a BIBD is also a PBD that is regular of degree $r$.
\end{remark}

Theorems~\ref{PBD theorem prot 5}~and~\ref{PBD theorem} demonstrate that the use of a regular PBD increases privacy against other users. This is because, from the point of view of another user, the possible source distribution is closer to uniform.

Theorem~\ref{PBD theorem prot 5} implies that for $U_t \in S_\ell$, if $U_t$ sees a query $q$ with proxy $U_j$ posted to $S_\ell$ that is not his own, any of the remaining $|S_\ell|-2$ users in $S_\ell$ are equally likely to be the source. If Protocol~\ref{covering design version 2} is used instead of Protocol~\ref{covering design version 1}, then $U_t$ can no longer completely eliminate the possibility of the proxy $U_j$ being the source. However, as Theorem~\ref{PBD theorem} shows, the likelihood of the proxy $U_j$ being the source is not the same as the likelihood of $U_i \neq U_j$ being the source. Indeed, it is far less likely that $U_j$ is acting as both proxy and source for $q$ in this situation. Intuitively, if a user $U_i$ is acting as both source and proxy, he has $r$ possible memory spaces to choose from, whereas if $U_i$ chooses another user $U_j$ as proxy, he has only $\lambda$ many memory spaces to choose from.

\subsection{Linked Queries and Coalitions of Users}

Users can also launch an intersection attack against a series of linked queries, similar to the intersection attack launched by DB against the DBWM and DBWMS protocols~(Protocols~\ref{DBWM protocol}~and~\ref{DBWMS Protocol}). The difference here is that users have access to the content of queries via the shared memory spaces; that is, users of a given memory space know which queries have been posted to that memory space, whereas the database only knows the identity of the proxy. 

\begin{example}
Consider the projective plane from Example~\ref{projective plane} and suppose we use Protocol~\ref{proxy-designated version 2}. Suppose that $U_1$ is the source of two linked queries, where the first query uses memory space $\{U_1, U_2, U_3\}$ and the second query uses memory space $\{U_1, U_4, U_5\}$. Now suppose that users $U_2$ and $U_5$ collude. From the first query, user $U_2$ knows that $U_i \in \{U_1,U_3\}$ (regardless of the proxy). From the second query, user $U_5$ knows that $U_i \in \{U_1,U_4\}$ (regardless of the proxy). If users $U_2$ and $U_5$ collude, then they can identify $U_1$ as the source.
\end{example}

In general, we can consider a sequence of $\rho$ linked queries made by the same (unknown) user, and a coalition $C$ of at most $c$ users that is trying to identify the source of the $\rho$ queries. We introduce the following terminology.

\begin{definition}
Consider a set of $\rho$ linked queries and fix a maximum coalition size $c$. If there are always at least $\kappa$ users who could possibly be the source (regardless of the queries and coalition) then we say that the scheme provides \emph{$(\rho,c,\kappa)$-anonymity}.
\end{definition}

\begin{remark}
Of course we want $\kappa \geq 2$ because the source might be identified if $\kappa=1$. 
\end{remark}

First, we consider security against a single user (i.e., the case $c= 1$). Here, it is advantageous to use a design with $\lambda = 1$: 

\begin{lemma}\label{single user} Suppose the BIBD chosen for Protocol~\ref{proxy-designated version 2} satisfies $\lambda = 1$. Then we achieve $(\rho,1,k-1)$-anonymity for any $\rho$.
\end{lemma}

\begin{proof}
If $U_i$ sees a sequence of $\rho$ linked queries from the same source,
the queries must all involve the same memory space, because $\lambda = 1$. The result then follows from Theorem~\ref{PBD theorem}.
\end{proof}

\begin{remark}The result of Lemma~\ref{single user} does not apply to Protocol~\ref{proxy-designated version 1}. This is because in Protocol~\ref{proxy-designated version 1}, given a series of linked queries posted to a given memory space, the only user who will never act as proxy for one of these queries is the query issuer. This is similar to the projective plane attack in~\cite{SB11} that we mentioned in Section~\ref{subsec: Attacks}.
\end{remark}

On the other hand, the security of Protocol~\ref{proxy-designated version 2} against a single user might be completely eliminated 
if we use a design with $\lambda > 1$. For example, suppose we use
a BIBD with $\lambda = 2$ in which every pair of blocks intersects in
at most two points (such a BIBD is termed \emph{supersimple}).
Consider two users $U_i$ and $U_j$. There exist two memory spaces, 
say $S_1$ and $S_2$, where $S_1 \cap S_2 = \{U_i,U_j\}$. 
Suppose $U_i$ observes two linked queries, say $q_1$ and $q_2$, that involve $S_1$ and $S_2$, respectively. Then $U_i$ can deduce that $U_j$ is the source.

We now consider some more special cases of this problem, for small values of $\rho$ and for certain special types of designs. This is because, in order to analyze the problem of linked queries, it becomes necessary to understand the block intersection properties of the scheme's chosen design. In Section~\ref{subsec: Methods to Increase Privacy}, we consider a more general approach to mitigate this type of attack in the Proxy-designated Covering Design Protocols~(Protocols~\ref{covering design version 1}~and~\ref{covering design version 2}).

The case $\rho = 1$ (i.e., security against a single query) is easy to analyze:

\begin{lemma} We achieve $(1,c,k-c-1)$-anonymity in Protocol~\ref{proxy-designated version 1}, where $c \leq k-3$. In Protocol~\ref{proxy-designated version 2}, we achieve $(1, c, k-c)$-anonymity, with the requirement that $c \leq k-2$.
\end{lemma}

\begin{proof}
We first consider Protocol~\ref{proxy-designated version 1}. Let $C$ be a coalition of size at most $c$ and let $S_h$ be the memory space used for the query $q_1$. Then $|C \cap S_h| \leq c$. $C$ can rule out as possible sources the users in $C \cap S_h$ as well as the proxy $U_j$ (provided that $U_j \not\in C \cap S_h$). Since $|S_h \backslash (C \cup \{U_j\}) | \geq k-c-1$, the result follows. An obvious requirement here is $c \leq k-3$.

For Protocol~\ref{proxy-designated version 2}, all other users with access to the given memory space can only eliminate themselves as the possible source of the query. This improves the information theoretic security for user-privacy against other users, as we now have $|S_h \backslash C | \geq k-c$. An obvious requirement here is $c \leq k-2$.
\end{proof}

For the case $\rho=2$, we consider BIBDs with a special intersection property.

\begin{lemma}\label{mu} Suppose the BIBD of Protocols~\ref{proxy-designated version 1}~and~\ref{proxy-designated version 2} satisfies the additional property that any two blocks intersect in at least $\mu$ points. Consider two linked queries, $q_1$ and $q_2$. Then we achieve $(2,c,\mu-c-2)$-anonymity, where $c \leq \mu - 4$, in Protocol~\ref{proxy-designated version 1}. In Protocol~\ref{proxy-designated version 2}, we achieve $(2, c, \mu-c)$-anonymity, with the requirement $c \leq \mu -2$.
\end{lemma}

\begin{proof}
Let $C$ be a coalition of size at most $c$ and let $S_{h_1}$ be the memory space used for the query $q_1$ and $S_{h_2}$ be the memory space used for $q_2$. Let $U_i$ be the proxy for $q_1$ and let $U_j$ be the proxy for $q_2$. 

In Protocol~\ref{proxy-designated version 1}, we have
\[|(S_{h_1} \backslash (C \cup \{U_i\})) \cap (S_{h_2} \backslash (C \cup \{U_j\}))| = |(S_{h_1}\cap S_{h_2}) \backslash (C \cup \{U_i,U_j\})| \geq \mu-c-2,\] so we achieve $(2,c,\mu-c-2)$-anonymity. An obvious requirement here is $c \leq \mu - 4$.

In Protocol~\ref{proxy-designated version 2}, we have  \[ |(S_{h_1} \backslash C)  \cap (S_{h_2} \backslash C)| = |(S_{h_1}\cap S_{h_2}) \backslash C | \geq \mu-c,\]
so we achieve $(2,c,\mu-c)$-anonymity. Here, an obvious requirement  is $c \leq \mu - 2$. 
\end{proof}

We can apply Lemma~\ref{mu} to the case of a symmetric BIBD, in which any two blocks intersect in exactly $\lambda$ points, as noted in Theorem~\ref{symmetric intersection}. This achieves the following result:

\begin{corollary}
Suppose the BIBD chosen for Protocols~\ref{proxy-designated version 1}~and~\ref{proxy-designated version 2} is a symmetric $(v, v, k, k, \lambda)$-BIBD. Then Protocol~\ref{proxy-designated version 1} provides $(2, c, \lambda - c - 2)$-anonymity for any $c \leq \lambda -4$ and Protocol~\ref{proxy-designated version 2} provides $(2, c, \lambda-c)$-anonymity for any $c \leq \lambda-2$.
\end{corollary}

An interesting extension to the concept of $(\rho,c,\kappa)$-anonymity is to consider an average-case analysis of privacy against other users. With $(\rho,c,\kappa)$-anonymity, we are analyzing the worst-case scenario---the minimum level of privacy the scheme achieves against \emph{any} possible coalition. While this is useful in some respects, schemes exhibiting powerful worst-case scenario attacks might actually perform quite well against a typical coalition. In particular, if a scheme needs to be concerned about random coalitions of users, such an average-case analysis might prove informative, as the following example shows.

\begin{example} Suppose we use a symmetric $(v, v, k, k, 3)$-BIBD in any of our P2P UPIR protocols. Consider linked queries $q_1$ and $q_2$ submitted by $U_i$, with corresponding memory spaces $S_{h_1}$ and $S_{h_2}$. By Theorem~\ref{symmetric intersection}, since the BIBD is symmetric, we have $|S_{h_1}\cap S_{h_2}| = 3$. That is, there are exactly two other users, say $U_j$ and $U_t$, in both $S_{h_1}$ and $S_{h_2}$. This implies that there is only \emph{one} coalition of users of size $2$ that can identify $U_i$ as the source. If we consider random coalitions, the probability that a random coalition of size 2 consists of $\{U_j,U_t\}$ is $\dfrac{1}{ {{v - 1} \choose 2 }}$. 

Let us consider other coalitions of size 2. Suppose $C = \{U_j, U_{\ell}\}$, for some user $U_\ell \neq U_t, U_i$. Then $C$ knows the source is either $U_t$ or $U_i$. There are $v-3$ such coalitions. The analysis for $C$ containing $U_t$ but not $U_j$ is similar. If we consider $C = \{U_\ell, U_{\ell'}\}$ such that $U_t, U_j \notin C$, the most advantageous coalition satisfies (without loss of generality) $U_\ell \in S_{h_1}$, $U_{\ell'} \in S_{h_2}$. In this case, $C$ sees both $q_1$ and $q_2$ and can conclude that the source is one of $\{U_i, U_j, U_t\}$. There are $(k-3)^2$ such coalitions. Other coalitions of size 2 either see only one of $\{q_1, q_2\}$, in which case the analysis reduces to that of Theorem~\ref{PBD theorem prot 5}~or~\ref{PBD theorem}, or neither of the linked queries, in which case $C$ can do nothing. 
\end{example}

\subsection{Some Methods to Increase Privacy}
\label{subsec: Methods to Increase Privacy}
\subsubsection{$t$-anonymity sets}
Beyond the limited cases described above, it is difficult to analyze the privacy guarantees of the proxy-designated BIBD and covering design protocols in the presence of linked queries. In particular, it becomes difficult to analyze the case of intersections of three or more memory spaces, and the size of these intersections probably decreases quickly. We might, however, wish to provide privacy for $\rho > 2$. One possible solution is to introduce the notion of built-in \emph{permanent anonymity sets} for each user. That is, suppose the set of users $\mathcal{U}$ is partitioned into \emph{anonymity sets} $\mathcal{T}_1, \ldots \mathcal{T}_g$, where each $\mathcal{T}_{\ell}$ consists of at least $t$ users. We further assume that the set system satisfies the property
 $\mathcal{T}_{\ell} \cap S_j \in \{\emptyset, \mathcal{T}_{\ell}\}$ for all $\ell, j$. We call such a construction a \emph{covering design with $t$-anonymity sets}.
 
\begin{theorem}Fix a partition $\mathcal{T} = \{\mathcal{T}_1, \ldots \mathcal{T}_g\}$ of the set of users $\mathcal{U}$, such that each $\mathcal{T}_{\ell}$ consists of at least $t$ users. Then we can construct a covering design with $t$-anonymity sets.
\end{theorem}

\begin{proof}
We can construct a covering design with $t$-anonymity sets by the following method. First, we construct a covering design on a set of $g$ points, say $\mathcal{X} = \{x_1, \ldots, x_g\}$. We then define a bijection $\sigma$ between the set of $g$ points and the $g$ anonymity sets, so $\sigma(\mathcal{X}) = \mathcal{T}$. Finally, for each $x_{\ell} \in \mathcal{X}$, we replace the point $x_{\ell}$ by the anonymity set $\sigma(x_{\ell}) = \mathcal{T}_{\ell'}$, where $1 \leq \ell' \leq g$. This yields a covering design satisfying the desired property.
\end{proof}

\begin{theorem} Fix a covering design with permanent anonymity sets of minimum size $t$. Then we achieve $(\rho, c, t-c-\rho)$-anonymity in Protocol~\ref{covering design version 1} and $(\rho, c, t-c)$-anonymity in Protocol~\ref{covering design version 2}:
\end{theorem}

\begin{proof}

Let $C$ be a coalition of size at most $c$ and consider a set of linked queries $q_1, \ldots, q_{\rho}$. Let $S_{h_{\ell}}$ be the memory space used for the query $q_{\ell}$ and let $U_{h_{\ell}}$ denote the proxy for $q_{\ell}$, for $1 \leq \ell \leq \rho$. 

In Protocol~\ref{covering design version 1}, we have
\begin{eqnarray*}
|(S_{h_1} \backslash (C \cup \{U_{h_1}\})) \cap (S_{h_2} \backslash (C \cup \{U_{h_2}\})) \cap \cdots \cap (S_{h_{\rho}} \backslash (C \cup \{U_{q_{\rho}}\}))|\\
 =
|(S_{h_1}\cap S_{h_2} \cap \cdots \cap S_{h_{\rho}}) \backslash (C \cup \{U_{h_1},U_{h_2}, \ldots, U_{h_{\rho}}\})| \geq t-c-\rho.
\end{eqnarray*}

In Protocol~\ref{covering design version 2}, we have
\begin{eqnarray*}
|(S_{h_1} \backslash C) \cap (S_{h_2} \backslash C) \cap \cdots \cap (S_{h_{\rho}} \backslash C )|
 =
|(S_{h_1}\cap S_{h_2} \cap \cdots \cap S_{h_{\rho}}) \backslash C| \geq t-c.
\end{eqnarray*} 
This completes the proof. \end{proof}

The idea of using permanent anonymity sets changes the trust requirements of the scheme. In particular, $U_i$ must trust the users contained in $\mathcal{T}_i$ to a greater extent than users in $\mathcal{U} \backslash \mathcal{T}_i$, since members of $\mathcal{T}_i$ necessarily have access to $U_i$'s query sphere. That is, there is no confidentiality among members of an anonymity set.

\subsubsection{Query hops}
\label{subsubsec: Query hops}
Another possible method to increase privacy against other users, which we briefly introduce here, involves the introduction of \emph{query hops} into the protocols. That is, we can consider allowing a designated proxy to rewrite a given query to another memory space, rather than simply forwarding the query to DB. We can establish a probabilistic approach, such that a designated proxy $U_j$ will, with some fixed probability $p$, forward the query to DB; otherwise $U_j$ rewrites the query uniformly at random to one of his associated memory spaces. When a response is received, a user simply posts the response back to the memory space where it was read from. This can continue until the query response reaches the source. In this case, it is easy to see that on average a query is posted $1/p$ times. This method removes the certainty a curious user has that the source of a given query is associated with the memory space in which that query is written. It is an interesting problem to analyze the privacy guarantees such a scheme provides against other users.

\section{Conclusion}
\label{sec: Conclusion}

In this paper, we have given an overview and analysis of current research in UPIR, including introducing an attack by the database on user privacy. We have established a new model for P2P UPIR and considered the problem of user privacy against other users in detail, going well beyond previous work. We have given two new P2P UPIR protocols and provided an analysis of the privacy properties provided by these protocols. Our P2P UPIR schemes, by taking advantage of the wide variety of available combinatorial designs, provide flexibility in the set-up phase, allowing for a choice between having a dynamic scheme (in which users are permitted to enter and leave the system), or providing increased privacy against other users. Finally, we have pointed out several directions for future research in this area. In particular, there is much work to be done regarding user privacy against other users, such as moving beyond the worst-case analysis we provide here and considering an average-case analysis, as well as the construction of P2P UPIR schemes that utilize query hops to mitigate loss of privacy against other users.

\section*{Acknowledgements} We would like to thank the
referees for their helpful remarks and suggestions.

\end{document}